\documentclass[10pt, onecolumn, journal]{IEEEtran}
\IEEEoverridecommandlockouts
\usepackage{hyperref}
\usepackage{tikz}
\usepackage[section]{placeins}
\usepackage{float}
\usepackage{amsmath}
\usepackage{graphicx}
\usepackage{epstopdf}

\usepackage{enumitem}
\usepackage{amssymb}
\usepackage{cite}
\usepackage{amssymb,amsmath}
\usepackage{amsgen,amsfonts,amsbsy,amsthm}
\usepackage{latexsym}
\usepackage{algorithm}
\usepackage[noend]{algpseudocode}
\makeatletter
\def\BState{\State\hskip-\ALG@thistlm}
\makeatother

\newcommand{\bvec}[1]{\mathbf{#1}}
\newcommand{\real}{\mathbb{R}}

\newcommand{\norm}[1]{\left\|#1\right\|_2}
\newcommand{\znorm}[1]{\left\|#1\right\|_0}
\newcommand{\opnorm}[1]{\left\|#1\right\|}
\newcommand{\zeroset}{\mathcal{C}_0}
\newcommand{\largeset}{\mathcal{C}_L}
\newcommand{\smallset}{\mathcal{C}_S}
\newcommand{\expect}{\mathbb{E}}
\newcommand{\sigpow}{P_x}
\newcommand{\noisepow}{P_{\nu}}
\newtheorem{lem}{Lemma}[section]
\newtheorem{thm}{Theorem}[section]
\newtheorem{cor}{Corollary}[section]

\begin{document}
\title{Performance Analysis of $l_0$ Norm Constrained Recursive Least Squares Algorithm}
 \author{Samrat~Mukhopadhyay, \IEEEmembership{Student Member, IEEE}, Bijit Kumar Das, \IEEEmembership{Student Member, IEEE}, and Mrityunjoy Chakraborty, \IEEEmembership{Senior Member, IEEE}%
  \thanks{Samrat Mukhopadhyay is with Dept.\ of Electronics and Electrical Communication Engineering, Indian Institute of Technology, Kharagpur 721302, India (Email: samratphysics@gmail.com)}%
  \thanks{Bijit Kumar Das is with Dept.\ of Electronics and Electrical Communication Engineering, Indian Institute of Technology, Kharagpur 721302, India (Email: bijitbijit@gmail.com)}%
  \thanks{Mrityunjoy Chakraborty is with Dept.\ of Electronics and Electrical Communication Engineering, Indian Institute of Technology, Kharagpur 721302, India
    (Email: mrityun@ece.iitkgp.ernet.in)}}
\maketitle
\begin{abstract}
Performance analysis of $l_0$ norm constrained Recursive least Squares (RLS) algorithm is attempted in this paper. Though the performance pretty attractive compared to its various alternatives, no thorough study of theoretical analysis has been performed. Like the popular $l_0$ Least Mean Squares (LMS) algorithm, in $l_0$ RLS, a $l_0$ norm penalty is added to provide zero tap attractions on the instantaneous filter taps. A thorough theoretical performance analysis has been conducted in this paper with white Gaussian input data under assumptions suitable for many practical scenarios. An expression for steady state MSD is derived and analyzed for variations of different sets of predefined variables. Also a Taylor series expansion based approximate linear evolution of the instantaneous MSD has been performed. Finally numerical simulations are carried out to corroborate the theoretical analysis and are shown to match well for a wide range of parameters.
\end{abstract}
\section{Introduction}
Sparse systems are frequently encountered in many applications, such as echo paths~\cite{duttweiler2000proportionate}, wireless communication channels, HDTV~\cite{schreiber1995advanced} etc. A system vector is called sparse if it has a very small number of nonzero entries compared to its dimension. It becomes necessary then to find identification algorithms suitable for such sparse systems. Adaptive algorithms are frequently used to identify systems whose parameters are changing with time. Due to its simplicity and ease of implementation, the least mean squares (LMS)  algorithm~\cite{widrow1985adaptive} has enjoyed much success for a long time. Another frequently used adaptive algorithm is the recursive least squares (RLS)~\cite{haykin2008adaptive} which recursively tries to minimize the error between estimated and unknown system vectors using the information conveyed by the data from the beginning of reception. But such algorithms are sparsity agnostic and generally do not perform well when the unknown system is sparse. Inspired by the introduction of sparse signal processing and the nascent field of Compressive sensing (CS)~\cite{candes2006robust,candes2006stable,candes2006near}, the last decade saw a flurry of activities on sparse adaptive filters, that has produced a number of several new algorithms that exploit the knowledge of sparsity~\cite{das2012sparse}. Many of these algorithms use the knowledge of sparsity of the unknown system vector to add an $l_p$ norm penalty to the cost function. ZA-LMS~\cite{chen2009sparse} uses $l_1$ norm penalty and $l_0$ LMS~\cite{gu2009norm} uses $l_0$ norm penalty to exert zero attraction on the filter taps. $l_1$ norm regularized RLS algorithms have also been proposed by researchers. The SPARLS~\cite{babadi2010sparls} algorithm suggests the use of Expectation-Maximization(EM) algorithm to minimize the $l_1$ norm penalized RLS cost function. The authors of~\cite{angelosante2010online} propose an algorithm that uses an online coordinate descent algorithm together with the $l_1$ regularized RLS cost function. The $l_1$ RLS algorithm~\cite{ekcsiouglu2010rls} has been proposed where the cost function of conventional RLS algorithm has been modified by adding a $l_1$ penalty term which results in a zero point attracted RLS algorithm. In~\cite{eksioglu2011rls} a general convex penalty term is added to the RLS cost function to result in a sparsity aware convex regularized RLS algorithm. 

~Among the different penalty terms that can be used as a regularizer of the cost function of RLS in~\cite{eksioglu2011rls}, of particular interests are the convex functions that can be used to approximate $l_0$ penalty term, as it was introduced in~\cite{gu2009norm}. Since the $l_0$ norm penalty can introduce strong zero point attraction to the small taps of the estimated parameter at each step of the algorithm, for a sparse system the algorithm is expected to converge faster to a lower steady state mean square deviation. Though the author of~\cite{eksioglu2011rls} has numerically shown that mean square deviation performance of $l_0$ norm penalized RLS is superior to the conventional RLS, neither he or anyone else, to the best of our knowledge, has been found to make an attempt to establish that claim through a theoretical analysis of the algorithm. A detailed theoretical analysis of such an algorithm could not only just corroborate the superior performance promised by the numerical simulations of $l_0$ RLS but also can find out the spectrum, of the different set of predefined variables, over which the algorithm may even become worse than the conventional algorithm. A detailed theoretical analysis of $l_0$ LMS was carried out in~\cite{su2012performance} which inspired the present work. The present work is aimed at providing a thorough analysis of the $l_0$ RLS algorithm along with presenting the salient features and limitations of the performance of this algorithm.

\section{Preliminaries}
Let the system has the unknown parameter vector $\bvec{s}=\begin{bmatrix}
s_0, s_1, \cdots s_{N-1}
\end{bmatrix}^T\in \real^{N}$ and let the input vector at time $n$ be denoted by $\bvec{x}_n=\begin{bmatrix}
x(n), x(n-1), \cdots x(n-N+1)
\end{bmatrix}^T\in \real^N$. The system produces output sequence $\{y(n)\}$ where \begin{align*}
y_n=\bvec{s}^T\bvec{x}_n+\nu_n
\end{align*}
where $\{\nu_n\}$ is an additive noise sequence. Let, the adaptive filter produces an estimate $\bvec{w}_n=\begin{bmatrix}
w_{0,n}, w_{1,n}, \cdots  w_{N-1,n}
\end{bmatrix}$ for the system tap vector, at time $n$.  The instantaneous estimation error between the output of the unknown system and the output of the adaptive filter is \begin{align*}
e_n=y_n-\bvec{w}_n^T\bvec{x}_n=(\bvec{s-w}_n)^T\bvec{x}_n+\nu_n
\end{align*}
The cost function of the conventional RLS adaptive filter with forgetting factor $\lambda$ is defined as \begin{align*}
\mathcal{E}_n=\sum_{m=0}^{n}\lambda^{n-m}(e_m)^2
\end{align*}
In order take into account the sparsity of the unknown system vector $\bvec{s}$, $l_0$-RLS modifies the cost function at each iteration by adding to it a penalty function that gives a measure of the sparsity of the system. $l_0$-RLS chooses the $l_0$ `'norm'` as the penalty function. As a result, the cost becomes, \begin{align}
\mathcal{E}_n=\sum_{m=0}^{n}\lambda^{n-m}(e_m)^2+\gamma\znorm{\bvec{w}_n}
\end{align} 
where the $l_0$ norm is defined as the number of non-zero entries of a vector and the parameter $\gamma$ is a penalty factor that controls the balance between estimation error and penalty. In general, the $l_0$ norm optimization problem is known to be NP hard~\cite{candes2006near} and because of that it is often approximated by continuous(often convex) functions. A popular approximation was introduced in~\cite{gu2009norm} which results, after some manipulations, in the following evolution equation of the $l_0$ RLS adaptive filter~\cite{eksioglu2011rls} \begin{align}
\label{eq:evolution}
\bvec{w}_{n}=w_{n-1}+\bvec{k}_n\xi_n+\kappa \bvec{P}_n\bvec{g}(\bvec{w}_{n-1})
\end{align} 
where \begin{align}
\label{eq:kapp-define}
\kappa=&\gamma(1-\lambda)\\
\label{eq:xi-define}
\xi_n=& y_n-\bvec{w}_{n-1}^T\bvec{x}_n\\
\label{eq:inverse-define}
\bvec{P}_n=& (\bvec{\Phi}_n)^{-1}=\left(\sum_{m=0}^n \lambda^{n-m}\bvec{x}_m\bvec{x}_m^T\right)^{-1}\\
\label{eq:gain-mat-define}
\bvec{k}_n=&\frac{\bvec{P}_{n-1}\bvec{x}_n}{\lambda+\bvec{x}_n^T\bvec{P}_{n-1}\bvec{x}_n}
\end{align}
and $\bvec{g}(\bvec{w}_{n-1})=\begin{bmatrix}
g(w_{0,n-1}),g(w_{1,n-1}),\cdots g(w_{N-1,n-1})
\end{bmatrix}^T$ where the function $g(\cdot)$ is defined as below \begin{align}
\label{eq:define-g-function}
g(t)=\left\{\begin{array}{lr}
\beta^2 t-\beta\mathrm{sgn}(t), & |t|\le 1/\beta\\
0, & \mathrm{elsewhere}
\end{array}\right.
\end{align}
The third term in Eq.~(\ref{eq:evolution}) is the \emph{zero-point attraction} term~\cite{gu2009norm} and the range $(-1/\beta,1/\beta)$ is called the \emph{attraction} range~\cite{su2012performance}.
\section{Modelling and assumptions}
\label{sec:assumptions}
Following the approach adopted by Su et.al~\cite{su2012performance}, based on the magnitudes of the entries of the unknown system vector $\bvec{s}$, we partition the set of indices $\{1,2,\cdots,\ N\}$ into three sets: \begin{align}
\label{eq:set-partition}
\zeroset:=&\{k:s_k=0\}\\
\largeset:=&\{k:|s_k|>1/\beta\}\\
\smallset:=&\{k:0<|s_k|\le 1/\beta\}
\end{align}
Thus, if $\bvec{s}$ is $K$-sparse, $|\largeset\cup\smallset|=K,\ |\zeroset|=N-K$.
 
We adopt the following assumptions:\begin{enumerate}[label=\bfseries {A.\arabic*}]
\item{\label{assumption:data}} The data sequence $\{x(n)\}$ is a white sequence with zero mean and variance $\sigpow$ and is independent of the additive noise sequence $\{\nu_n\}$ which is also assumed to be a zero mean sequence.
\item{\label{assumption:independence}}(\textit{Independence assumption}) The incoming sequence of vectors $\{\bvec{x}_n\}$ are independent.
\item{\label{assumption:inverse-simplification}} $\lambda$ is chosen \emph{sufficiently} close to $1$ such that $\frac{N}{N+1}<<\lambda<1$, so that for large $n$, $\bvec{P}_n\approx \mathbb{E}(\mathbb{P}_n)=\frac{1-\lambda}{1-\lambda^{n+1}}\bvec{R}^{-1}$ where $\bvec{R}$ is the autocorrelation matrix of the incoming data sequence.
\item{\label{assumption:kappa-param}} The parameters $\kappa$ and $\beta$ are chosen such that $\beta^2\kappa(1-\lambda)<<\sigpow$.
\item{\label{assumption:zero-tap-weight}} The tap weights $w_{k,n},\ \forall k\in \zeroset$ are gaussian distributed. 
\item{\label{assumption:tap-sign}} $w_{k,n}$ are assumed to be of the same sign as that of $s_k$, $k\in \largeset\cup \smallset$.
\item{\label{assumption:attraction}} $w_{k,n}$ are assumed to be out of the attraction range for $k\in \largeset$ and inside attraction range elsewhere.
\end{enumerate}
The following points attempt to justify the use of these assumptions:
\begin{enumerate}
\item The assumption~\ref{assumption:data} is generally adopted to leverage the simple properties of a gaussian data sequence. This assumption can be slightly generalized by dropping the assumption that the sequence is independent, which forces one to work with a coloured gaussian sequence. However, a coloured sequence be easily \emph{pre-whitened} by pre-multiplying any vector of interest with the unitary matrix that diagonalizes the covariance matrix of the gaussian sequence~\cite{haykin2008adaptive}, which is why assumption~\ref{assumption:data} can be considered without loss of generality.
\item Assumption~\ref{assumption:independence} is the \emph{independence} assumption and is widely used in the literature for simplified analysis of adaptive algorithms~\cite{haykin2008adaptive},~\cite{kushner2003stochastic}.
\item  Assumption~\ref{assumption:inverse-simplification} has generally been used in the literature for simplified analysis of RLS~\cite{uncini2015fundamentals}. One justification for this assumption can be provided by the following lemma which assumes assumptions~\ref{assumption:data}~and~\ref{assumption:independence}. 
\begin{lem}
\label{lem:ergodic-type}
If the sequence $\{x(n)\}$ is assumed to follow assumption~\ref{assumption:data}, then with $0<\lambda\le 1$ \begin{align}
\label{eq:ergodic-type}
\lim_{n\to \infty} \expect\left(\opnorm{\frac{\bvec{\Phi}_n}{\frac{1-\lambda^{n+1}}{1-\lambda}}-\bvec{R}}^2\right)=\left(\frac{1-\lambda}{1+\lambda}\right)(N+1)\sigpow^2
\end{align}
where $\opnorm{\cdot}$ denotes the $2$-matrix norm. 
\end{lem}
\begin{proof}
A short proof is provided in Appendix~\ref{sec:appendix-ergodic-type}.
\end{proof}
Lemma~\ref{lem:ergodic-type} encourages the use of assumption~\ref{assumption:inverse-simplification}. Furthermore, as it will be seen in the performance analysis of $l_0$ RLS, this assumption simplifies the analysis significantly because without this assumption, the nonlinear contribution of past data vector $\bvec{x}_{n-1}$, in matrix $\bvec{P}_n$ makes carrying out the analysis difficult.
\item Assumption~\ref{assumption:kappa-param} is a result of experimental observation. It basically implies that for the $l_0$ RLS to be stable, $\kappa,\beta,1-\lambda$ have to be small compared to the signal power.
\item The use of assumptions~\ref{assumption:zero-tap-weight},~\ref{assumption:tap-sign}, and~\ref{assumption:attraction} are found suitable for this analysis. These exactly same assumptions are
taken in\cite{su2012performance} for the analysis of $l_0$ LMS. They justifications of the assumptions there are based on intuitive discussion and logical assumptions which also were probably justified by experimental observations. In the same spirit we also performed extensive simulations to verify these assumptions. Also, since the structure of the $l_0$ RLS algorithm is similar to that of the $l_0$ LSM algorithm, save the time varying gain matrix, it is expected that the logical discussions similar to those justifying the use of these assumptions in the work of Su etal~\cite{su2012performance} can also justify the use of these assumptions in our work.  
\end{enumerate}
 
\section{Performance analysis}
The convergence analysis of RLS itself is not easy because of the presence of the time dependent gain matrix $\bvec{P}_n$. However the use of assumption~\ref{assumption:inverse-simplification} significantly simplifies the analysis~\cite{uncini2015fundamentals}. We then use the assumptions taken in Section~\ref{sec:assumptions} to carry out the analysis in a simplified manner. 
\subsection{Mean convergence analysis:}
\label{sec:mean-convergence}  
Define $\bvec{h}_n=\bvec{w}_n-\bvec{s}$ as the weight deviation vector. Recalling the equation of evolution for the adaptive filter from Eq.~(\ref{eq:evolution}), the recursive update equation for $\bvec{h}(n)$ can be written as \begin{align*}
\bvec{h}_{n}=(\bvec{I}-\bvec{k}_n\bvec{x}_n^T)\bvec{h}_{n-1}+\bvec{k}_n\nu_n+\kappa\bvec{P}_n\bvec{g}(\bvec{w}_{n-1})
\end{align*}
where the definition of $\bvec{\xi}_n$ from Eq.~(\ref{eq:xi-define}) and the equation for $y_n$ have been evoked. The sequence of inverse matrices $\{\bvec{P}_n\}$ evolve according to the following well known \emph{Riccati} equation~\cite{uncini2015fundamentals}\begin{align}
\label{eq:inverse-evolution}
\bvec{P}_n=\lambda^{-1}(\bvec{I}-\bvec{k}_n\bvec{x}_n^T)\bvec{P}_{n-1}
\end{align}
Using this update quation of $\bvec{P}_n$, the filter evolution equation takes the form\begin{align}
\label{eq:main-evolution}
\bvec{h}_{n}=\lambda \bvec{P}_n\bvec{\Phi}_{n-1}\bvec{h}_{n-1}+\bvec{k}_n\nu_n+\kappa\bvec{P}_n\bvec{g}(\bvec{w}_{n-1})
\end{align}
Utilizing assumptions~\ref{assumption:inverse-simplification} and~\ref{assumption:data}, we can further simplify the evolution equation to get (for large $n$) \begin{align}
\label{eq:simplified-evolution}
\bvec{h}_{n}=\eta_n\bvec{h}_{n-1}+\bvec{k}_n\nu_n+\rho_n\bvec{g}(\bvec{w}_{n-1})
\end{align}
where the following symbols are be used to compactly represent the expressions that will be derived in the paper:
\begin{align}
\label{eq:eta-define}
\eta_n=&\lambda\left(\frac{1-\lambda^n}{1-\lambda^{n+1}}\right)\\
\label{eq:rho-define}
\rho_n=&\frac{\kappa}{\sigpow}\left(\frac{1-\lambda}{1-\lambda^{n+1}}\right)\\
\label{eq:cn-define}
c_n=&\lambda^n\left(\frac{1-\lambda}{1-\lambda^{n+1}}\right)\\
\label{eq:dn-define}
d_n=&\frac{\kappa}{\sigpow}\left(\frac{1-\lambda^n}{1-\lambda^{n+1}}\right)\\
\label{eq:theta-define}
\theta=&\frac{\beta\kappa(1-\lambda)}{\sigpow}
\end{align}  
Then the following theorem describes the evolution and convergence of the mean of deviation vector $\bvec{h}_n$.
\begin{thm}
\label{thm:mean-convergence}
The mean deviation coordinates $\expect h_{k,n}$ evolve according to the following recursive equation 
\begin{align}
\label{eq:mean-evolution-solved}
\expect h_{k,n}=\left\{\begin{array}{ll}
c_n\expect h_{k,0}+d_ng(s_k),& k\in \smallset\\
c_n\expect h_{k,0},& k\in \zeroset\cup \largeset
\end{array}
\right.
\end{align}
As a result, \begin{align}
\label{eq:mean-convergence-equation}
\expect h_{k,\infty}=\left\{\begin{array}{ll}
\frac{\kappa}{\sigpow} g(s_k),& k\in \smallset\\
0, & k\in \zeroset\cup \largeset
\end{array}
\right.
\end{align}
\end{thm}
\begin{proof}
The proof is postponed to Appendix~\ref{sec:appendix-mean-convergence}.
\end{proof}
\subsection{Mean Square convergence analysis:}
 We begin by investigating the evolution of the correlation matrix of the mean deviation vector, i.e. $\expect \bvec{h}_n\bvec{h}_n^T$. From Eq.~(\ref{eq:simplified-evolution}) we get \begin{align}
\label{eq:tempo-mean-square-analysis}
\expect\bvec{h}_n\bvec{h}_n^T & =\bvec{M}_1+(\bvec{M}_2+\bvec{M}_2^T)+(\bvec{M}_3+\bvec{M}_3^T)\\
\ & +(\bvec{M}_4+\bvec{M}_4^T)+\bvec{M}_5+\bvec{M}_6
\end{align}
where \begin{align}
\label{eq:M1}
\bvec{M}_1=& \lambda^2\expect\left(\bvec{P}_n\bvec{\Phi}_{n-1}\bvec{h}_{n-1}\bvec{h}_{n-1}^T\bvec{\Phi}_{n-1}\bvec{P}_n\right)\\
\bvec{M}_2=&\lambda \expect\left(\bvec{P}_n\bvec{\Phi}_{n-1}\bvec{h}_{n-1}\bvec{k}_n^T\nu_n\right)\\
\bvec{M}_3=&\lambda\kappa\expect\left(\bvec{P}_n\bvec{\Phi}_{n-1}\bvec{h}_{n-1}\bvec{g}^T(\bvec{w}_{n-1})\bvec{P}_n\right)\\
\bvec{M}_4=&\kappa\expect\left(\bvec{P}_{n}\bvec{g}(\bvec{w}_{n-1})\bvec{k}_n^T\nu_n\right)\\
\bvec{M}_5=&\kappa^2\expect\left(\bvec{P}_n\bvec{g}(\bvec{w}_{n-1})\bvec{g}(\bvec{w}_{n-1})^T\bvec{P}_n\right)\\
\bvec{M}_6=&\expect\left(\nu_n^2 \bvec{k}_n\bvec{k}_n^T\right)
\end{align}
By using assumptions~\ref{assumption:data},~\ref{assumption:independence}, and~\ref{assumption:inverse-simplification}, we get the following simplified equations for the terms in the right hand side of ~(\ref{eq:tempo-mean-square-analysis}):
\begin{align}
\bvec{M}_1=&\eta_n^2\expect\bvec{h}_{n-1}\bvec{h}_{n-1}^T\\
\bvec{M}_2=&\bvec{0}\\
\bvec{M}_3=&\eta_n\rho_n\left(\frac{1-\lambda^n}{1-\lambda^{n+1}}\right)\expect\left(\bvec{h}_{n-1}\bvec{g}^T(\bvec{w}_{n-1})\right)\\
\bvec{M}_4=&\bvec{0}\\
\bvec{M}_5=&\rho_n^2\expect\left(\bvec{g}(\bvec{w}_{n-1})\bvec{g}(\bvec{w}_{n-1})^T\right)\\
\bvec{M}_6=&\noisepow\expect\bvec{k}_n\bvec{k}_n^T
\end{align}
Thus, the evolution equation for the correlation matrix of $\bvec{h}_n$ can be expressed as \begin{align}
\label{eq:evolve-error-covariance-mat}
\expect\bvec{h}_n\bvec{h}_n^T=&\eta_n^2\expect\bvec{h}_{n-1}\bvec{h}_{n-1}^T+\eta_n\rho_n
\expect\left(\bvec{h}_{n-1}\bvec{g}^T(\bvec{w}_{n-1})+\bvec{g}(\bvec{w}_{n-1})\bvec{h}^T_{n-1}\right)+\rho_n^2\expect\left(\bvec{g}(\bvec{w}_{n-1})\bvec{g}(\bvec{w}_{n-1})^T\right)+\noisepow\expect\bvec{k}_n\bvec{k}_n^T
\end{align} 
Taking the $k^\mathrm{th}$ diagonal element of the error covariance matrix we get the corresponding evolution equation:\begin{align}
\label{eq:evolve-diag-error-covariance-mat}
\expect h_{k,n}^2=\eta_n^2\expect h_{k,n-1}^2+2\eta_n\rho_n\expect(h_{k,n-1}g(w_{k,n-1}))
+\rho_n^2\expect(g^2(w_{k,n-1}))+\noisepow\expect(k_{k,n})^2
\end{align}

 To do the mean square convergence analysis, we introduce the notations that will be henceforth used to succinctly represent the results of the mean square convergence analysis.\begin{align}
\label{eq:Dn-define}
D_n:=&\expect \norm{\bvec{h}_n}^2\\
\label{eq:Omegan-define}
\Omega_n:=&\sum_{k\in \zeroset} \expect h^2_{k,n}\\
\label{eq:small-omegan-define}
\omega_n^2:=&\expect h^2_{k,n}\quad \forall\ k\in \zeroset\\
\label{eq:G(s)-define}
G(s):=&\sum_{k\in \zeroset}g^2(s_k)\\
\label{eq:G'(s)-define}
G'(s):=&\sum_{k\in \zeroset}s_kg(s_k)
\end{align}
\subsubsection{\textbf{Instantaneous approximate mean square deviation analysis}}
In this section we provide the result of an approximate analysis for the instantaneous MSD.  
\begin{thm}
\label{thm:evolve-instantaneous-zero-nonzero-tap-pow}
The instantaneous power of the nonzero and zero taps of the $l_0$ RLS filter evolve, approximately, according to the following linear dynamical system:
\begin{align}
\label{eq:evolve-instantaneous-zero-nonzero-tap-pow}
\begin{bmatrix}
D_n\\
\Omega_n
\end{bmatrix}=\bvec{A}_n\begin{bmatrix}
D_{n-1}\\
\Omega_{n-1}
\end{bmatrix}+\bvec{b}_{n}
\end{align}
where \begin{align}
\label{eq:instantaneous-evolve-matrix}
\bvec{A}_n=\begin{bmatrix}
\eta_n^2 & -\frac{2\beta\rho_n\eta_n}{\sqrt{2\pi\omega_{\infty}^2}}\\
0 & \eta_n^2-\frac{2\beta\rho_n\eta_n}{\sqrt{2\pi\omega_{\infty}^2}}
\end{bmatrix}
\end{align}
and \begin{align}
\label{eq:instantaneous-evolve-vec}
\bvec{b}_n=\begin{bmatrix}
b_n(1)\\
b_n(2)
\end{bmatrix}
\end{align}
where \begin{align*}
b_n(1)=&N\noisepow p_n^2+(N-K)\beta^2\rho_n^2-2(N-K)\beta\rho_n\eta_n\omega_{\infty}^2/\sqrt{2\pi\omega_{\infty}^2}\\
\ &-2\rho_n c_n\eta_n G'(s)+(2\rho_n d_n\eta_n+\rho_n^2)G(s)\\
b_n(2)=&(N-K)(\noisepow p_n^2+\beta^2\rho_n^2)-2(N-K)\beta\rho_n\eta_n\omega_{\infty}^2/\sqrt{2\pi\omega_{\infty}^2}
\end{align*} 
and,  \begin{align}
\label{eq:steady-zerotap}
\omega_{\infty}=&\frac{-2\lambda\theta/\sqrt{2\pi}+
\sqrt{2\lambda^2\theta^2/\pi+(1-\lambda^2)(\theta^2+\noisepow p_\infty^2)}}{1-\lambda^2}
\end{align}
where $\theta$ is defined as in Equation~\ref{eq:theta-define}.
\end{thm}
\begin{proof}
The proof is postponed to Appendix~\ref{sec:appendix-evolve-instantaneous-zero-nonzero-tap-pow}.
\end{proof}
\subsubsection{\textbf{Steady state mean square deviation analysis}}
Unlike the instantaneous analysis, we can get the expression for steady state MSD exactly under the assumptions taken in Sec.~\ref{sec:assumptions}. The result of that analysis is showed in the form of the following theorem.
\begin{thm}
\label{thm:steady-msd}
The steady state MSD has the following expression:
\begin{align}
\label{eq:steady-msd}
D_\infty=\frac{N\noisepow p_\infty^2}{1-\lambda^2}+\beta_1\theta^2-\beta_2\theta\sqrt{\theta^2+\beta
_3}
\end{align}
where \begin{align*}
\beta_1:=&\frac{N-K}{1-\lambda^2}+\frac{G(s)}{\beta^2(1-\lambda)^2}+\frac{4\lambda^2(N-K)}{\pi(1-\lambda^2)^2}\\
\beta_2:=&\frac{4\lambda(N-K)}{\sqrt{2\pi}(1-\lambda^2)^2}\sqrt{\frac{2\lambda^2}{\pi}+1-\lambda^2}\\
\beta_3:=&\frac{\noisepow p_\infty^2}{\frac{2\lambda^2}{\pi(1-\lambda^2)}+1}
\end{align*}
\end{thm}
\begin{proof}
The proof is postponed to Appendix~\ref{sec:appendix-steady-msd}.
\end{proof}
The appearance of the form of the steady state MSD is identical to the one derived by the authors of~\cite{su2012performance} since our analysis actually follows the same methodology as theirs. But the terms that calculate the MSD are quite different and also the way the terms $\beta_1,\beta_2,\beta_3$ depend upon the attraction parameter $\beta$ is different from the way the dependence is for $l_0$ LMS (See~\cite{su2012performance} for details). The first term in Eq.~(\ref{eq:steady-msd}) is the steady state MSD for conventional RLS and the second and third terms comprise of the ``excess'' MSD produced by the $l_0$ attraction term. Note that this excess MSD can very well be negative, for certain range of $\kappa$, which results in improved performance of $l_0$ RLS. In fact, paralleling Corollary 1 of~\cite{su2012performance}, we can get the following corollaries from straight forward calculations:
\begin{cor}
\label{cor:kappa-range}
For fixed $\beta$, $l_0$ RLS outperforms conventional RLS if the parameter $\kappa$ is chosen such that the following holds \begin{align}
\label{eq:kappa-range}
0<\theta<\sqrt{\frac{\beta_2^2\beta_3}{\beta_1^2-\beta_2^2}}
\end{align} 
\end{cor}
\begin{proof}
The proof follows by noticing that $l_0$ RLS outperforms conventional RLS in steady state MSD if $D_{\infty}<\frac{N\noisepow p_{\infty}^2}{(1-\lambda^2)}\implies \beta_1\theta^2-\beta_2\theta\sqrt{\beta_3+\theta^2}<0$ and recalling that $\theta>0$.
\end{proof}
\begin{cor}
\label{cor:kappa-opt}
In terms of minimum obtainable MSD from $l_0$ RLS, the best choice of $\kappa$ is found from 
\begin{align}
\label{eq:kappa-opt}
\theta_{\mathrm{opt}}=\frac{\sqrt{\beta_3}}{2}\left(\sqrt[4]{\frac{\beta_1+\beta_2}{\beta_1-\beta_2}}-\sqrt[4]{\frac{\beta_1-\beta_2}{\beta_1+\beta_2}}\right)
\end{align}
and the minimum MSD is \begin{align}
\label{eq:min-msd}
D_{\infty}^{\mathrm{min}}=\frac{N\noisepow p_{\infty}^2}{1-\lambda^2}+\frac{\beta_3}{2}\left(\sqrt{\beta_1^2-\beta_2^2}-\beta_1\right)
\end{align}
\end{cor}
\begin{proof}
The proof is the same as the proof of Corollary 1 in~\cite{su2012performance}. The readers are referred to Appendix A of~\cite{su2012performance} for details.
\end{proof}
From the definitions of $\beta_1,\beta_2,\beta_3$ in Theorem~\ref{thm:steady-msd}, it is evident from Corollary~\ref{cor:kappa-opt} that the minimum MSD given by $l_0$ RLS is a function of the attraction parameter $\beta$. The following corollary shows that this minimum MSD is, in fact, constant if $\beta$ is large. \begin{cor}
\label{cor:min-msd-vs-beta}
The minimum steady state MSD $D_{\infty}^{\mathrm{min}}$ is a decreasing function of $\beta$ and as $\beta\to \infty$, the ratio of minimum MSD of $l_0$ RLS, as found in Corollary~\ref{cor:kappa-opt} and the steady state MSD of conventional RLS converges to \begin{align}
\label{eq:min-msd-vs-beta}
\lim_{\beta\to \infty}\frac{D_{\infty}^{\mathrm{min}}}{D_{\mathrm{RLS}}}=&\frac{\pi(1-\lambda^2)+2\frac{K}{N}\lambda^2}{\pi(1-\lambda^2)+2\lambda^2}
\end{align}
which is $\approx K/N$ if $\lambda$ is close to $1$.
\end{cor}
\begin{proof}
First, observe that $\beta_2,\beta_3$ are independent of $\beta$ and the only dependence of the steady state MSD on $\beta$ is through the term $\beta_1$. From Equation~\ref{eq:steady-msd} it is clear that the steady state MSD is an increasing function of $\beta_1$ and from the expression of $\beta_1$ it is clear that $\beta_1$ is a decreasing function of $\beta$, which proves the first part of the corollary.

 To see how the second part of the corollary comes up, observe that the the expression for $\beta_1$ can be rewritten from Theorem~\ref{thm:steady-msd} as \begin{align*}
\beta_1=&\frac{G(s)}{\beta^2(1-\lambda)^2}+\frac{N-K}{(1-\lambda^2)^2}\left(\frac{4\lambda^2}{\pi}+1-\lambda^2\right)\\
\implies\lim_{\beta\to \infty}\beta_1 =&\frac{N-K}{(1-\lambda^2)^2}\left(\frac{4\lambda^2}{\pi}+1-\lambda^2\right)
\end{align*}
Now, to make the expressions look less formidable, let $$f_0=\noisepow p_{\infty}^2,\ f_1=\frac{(N-K)}{(1-\lambda^2)^2},\ f_2^2=\frac{2\lambda^2}{\pi}+1-\lambda^2,\ f_3^2=\frac{2\lambda^2}{\pi}$$ then, $$\lim_{\beta\to \infty}\beta_1=f_1(f_2^2+f_3^2),\ \beta_2=2f_1f_2f_3,\ \beta_3=\frac{f_0(1-\lambda^2)}{f_2^2}$$
so that \begin{align*}
\lim_{\beta\to\infty}\lefteqn{D_{\infty}^{\mathrm{min}}}\\
= & \displaystyle\frac{(1-\lambda^2)Nf_0f_1}{N-K}+\frac{\beta_3}{2}\left(\sqrt{\beta_1^2-\beta_2^2}-\beta_1\right)\\
= & \displaystyle\frac{(1-\lambda^2)Nf_0f_1}{N-K}+\frac{f_0f_1(1-\lambda^2)}{2f_2^2}\left(\sqrt{(f_2^2+f_3^2)^2-4f_2^2f_3^2}-(f_2^2+f_3^2)\right)\\
= & \displaystyle\frac{(1-\lambda^2)Nf_0f_1}{N-K}-\frac{(1-\lambda^2)f_0f_1f_3^2}{f_2^2}\\
= & D_{RLS}\displaystyle \left(1-\frac{(N-K)f_3^2}{Nf_2^2}\right)
\end{align*}
from where the result follows after plugging in the expressions for $f_2^2$ and $f_3^2$.
\end{proof}
Another important observation is that the expression of minimum steady state MSD in Equation~\ref{eq:min-msd} is dependent upon the unknown system parameters in the set $\smallset$. This dependence is via $G(s)$ which appears in the expression of $\beta_1$. Interestingly, the extent of this dependence is controlled by the attraction parameter $\beta$, and as seen from the Corollary~\ref{cor:min-msd-vs-beta}, this dependence vanishes when $\beta$ becomes large and then the MSD is only a function of $\lambda$ and the system sparsity to length ratio $K/N$. In this regard, the following simple corollary connects the behaviour of the minimum steady state MSD with the sparsity of the system and the attraction of the small unknown parameters $G(s)$.
\begin{cor}
\label{cor:msd-vs-sparsity}
The minimum steady state MSD in Eq.~(\ref{eq:min-msd}) is a monotonically increasing function of the small set attraction $G(s)$ and the sparsity $K$. 
\end{cor}
\begin{proof}
We can write the expression for the minimum steady state MSD as \begin{align*}
D_{\infty}^{\mathrm{min}}=\frac{Np_{\infty^2}}{(1-\lambda^2)}-\frac{\beta_2^2\beta_3}{\sqrt{\beta_1^2-\beta_2^2}+\beta_1}
\end{align*}
which shows that $D_{\infty}^{\mathrm{min}}$ increases with the increase of $\beta_1$. Then, as $\beta_1$ is an increasing function of $G(s)$, $D_{\infty}^{\mathrm{min}}$ is also a monotonically increasing function of $G(s)$.

To investigate the dependence of the minimum steady state MSD on the sparsity $K$, first note that the first term is independent of $K$ and hence the behaviour of the second term will suffice for our purpose. 
Now, let us define, for the sake of simplicity of the expressions, $$\ f_1=\frac{(N-K)}{(1-\lambda^2)^2},\ f_2^2=\frac{2\lambda^2}{\pi}+1-\lambda^2,\ f_3^2=\frac{2\lambda^2}{\pi}, f_4=\frac{G(s)}{\beta^2(1-\lambda)^2},$$ then, $$\beta_1=f_1(f_2^2+f_3^2)+f_4,\ \beta_2=2f_1f_2f_3$$ Then, note that we can express the second term as a function of $f_1$ (and hence as a function of $N-K$) in the following manner:
\begin{align*}
\frac{\beta_3}{2}\lefteqn{\left(\sqrt{\beta_1^2-\beta_2^2}-\beta_1\right)}\\
& =\frac{\beta_3}{2}\left(\sqrt{f_1^2(f_2^2-f_3^2)^2+2f_1f_4(f_2^2+f_3^2)+f_4^2}-(f_1(f_2^2+f_3^2)+f_4)\right)\\
& = \frac{-2\beta_3f_1^2f_2^2f_3^2}{\sqrt{f_1^2(f_2^2f_3^2)^2+2f_1f_4(f_2^2+f_3^2)+f_4^2}+(f_1(f_2^2+f_3^2)+f_4)}\\
& = \frac{E_1}{E_2}
\end{align*}
It is trivial to note that $E_1$ is negative and decreases as $f_1$ increases. In the same way it is easy to verify that $E_2$ is positive and increases with $f_1$. Thus, the second term decreases as $f_1$ increases, which implies, that the second term increases when $K$ increases. This proves that the minimum steady state MSD increases with the increase in sparsity $K$. 
\end{proof}   
\section{Numerical Experiments}
Numerical experiments are carried out to verify the accuracy of our analysis. In order to perform the experiments, the unknown system vector $\bvec{s}$, is generated by generating its components as independent samples of a $\mathcal{N}(0,1)$ random variable. Each simulation result is averaged over $100$ iterations. Table~\ref{tab:experiments} documents the various parameter values that are used during the experiments. 

Figures~\ref{fig:exp1-50dB} and~\ref{fig:exp1-25dB} compare the steady state MSD of conventional RLS, MSD of $l_0$ RLS obtained from simulation and MSD of $l_0$ RLS obtained from the analysis that resulted in Eq.~(\ref{eq:steady-msd}) as $\kappa$ is varied. The figure clearly shows that the theory is in good agreement with the simulation. Also, the value of the optimal $\kappa_{\mathrm{opt}}$ is seen to be well matched with that found from simulation. It can be seen that the tally is better when SNR is $50$ dB than when SNR is $25$dB. This is expected since decrease in SNR makes the assumptions~\ref{assumption:zero-tap-weight} and~\ref{assumption:attraction} weak. 

Figures~\ref{fig:exp2-50dB} and~\ref{fig:exp2-25dB} plot the variation of steady state MSD with $\beta$. it can be seen that the result from analysis matches well with the theory. Also, it is interesting to observe that the decrease in the MSD for $l_0$ RLS is almost by a factor of $1/10$ compared to the steady state MSD of conventional RLS. This result matches quite closely with the result stated in Corollary~\ref{cor:min-msd-vs-beta}, according to which, this factor should be$\approx K/N=1/10.667$ using the values of $K,N$ from Table.~\ref{tab:experiments} for experiment 2.   

Figure~\ref{fig:exp3-50dB} plots the variation of steady state MSD with sparsity $K$. The figure clearly verifies the claim of Corollary~\ref{cor:msd-vs-sparsity}.
\begin{center}
\begin{table}[t!]
\caption{\textsc{Parameter values for different numerical experiments}}
\label{tab:experiments}
\begin{tabular*}{\textwidth}{@{\extracolsep{\fill}}ccccccc}
\hline
Experiment & $N$ & $K$ & $\lambda$ & $\beta$ & $\kappa$ & SNR\\
\hline
$1$ & $64$ & $6$ & $0.995$ & $5$ & $5\times 10^{-7}\to 10^{0.1}\kappa_{\mathrm{max}}$ & 50dB/25dB\\
\hline
$2$ & $64$ & $6$ & $0.995$ & $10^{-1}\to 50$ & $\kappa_{\mathrm{opt}}$ & 50dB\\
\hline
$3$ & $64$ & $1\to 61$ & $0.995$ & $5$ & $\kappa_\mathrm{opt}$ & 50dB\\
\hline
\end{tabular*}
\end{table} 
\end{center}
\begin{figure}[t!]
\includegraphics[height=3in,width=6in]{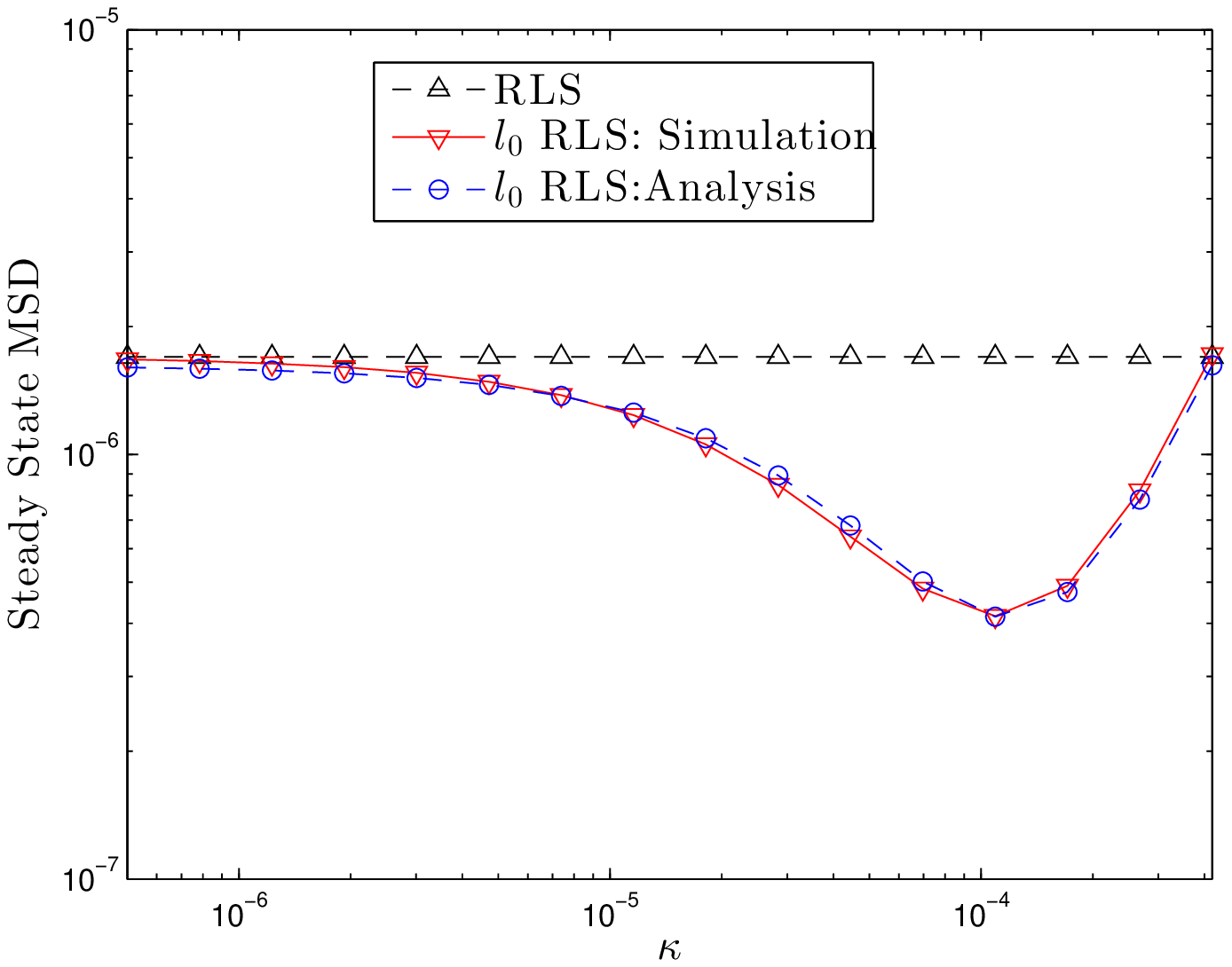}
\caption{Steady state MSD vs $\kappa$ for SNR 50dB}
\label{fig:exp1-50dB}
\end{figure}
\begin{figure}[t!]
\includegraphics[height=3in,width=6in]{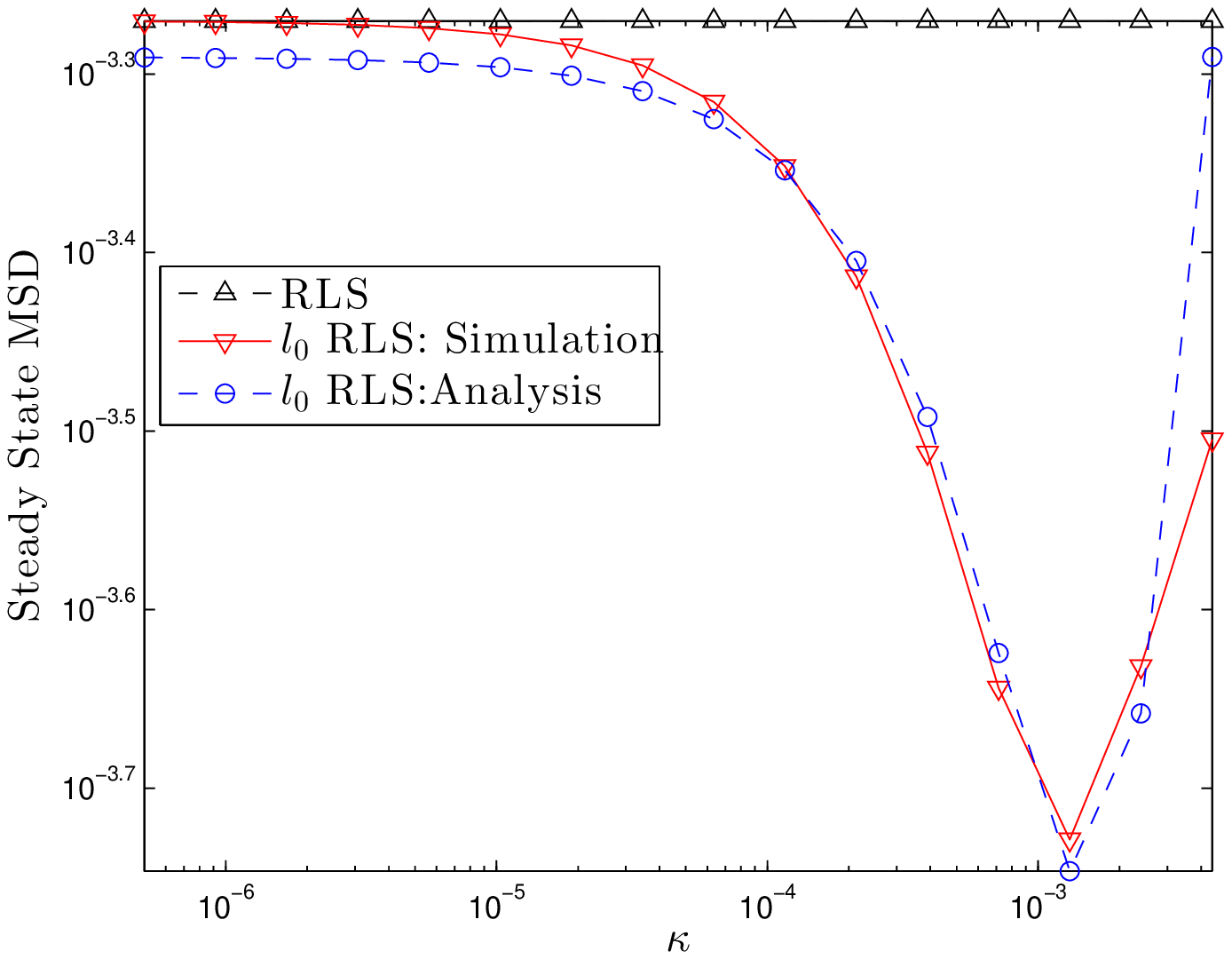}
\caption{Steady state MSD vs $\kappa$ for SNR 25dB}
\label{fig:exp1-25dB}
\end{figure}
\begin{figure}[t!]
\includegraphics[height=3in,width=6in]{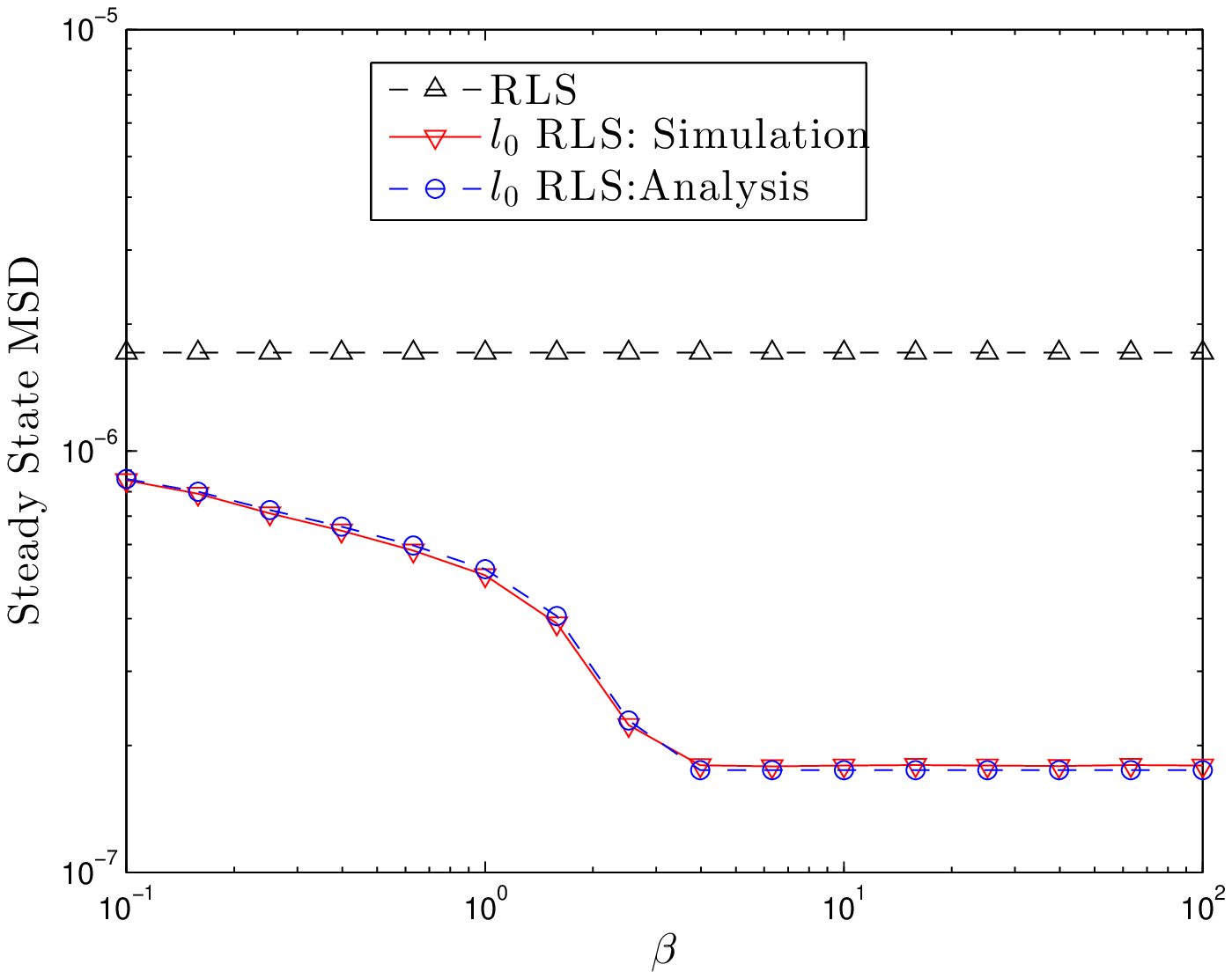}
\caption{Steady state MSD vs $\beta$ for SNR 50dB}
\label{fig:exp2-50dB}
\end{figure}
\begin{figure}[t!]
\includegraphics[height=3in,width=6in]{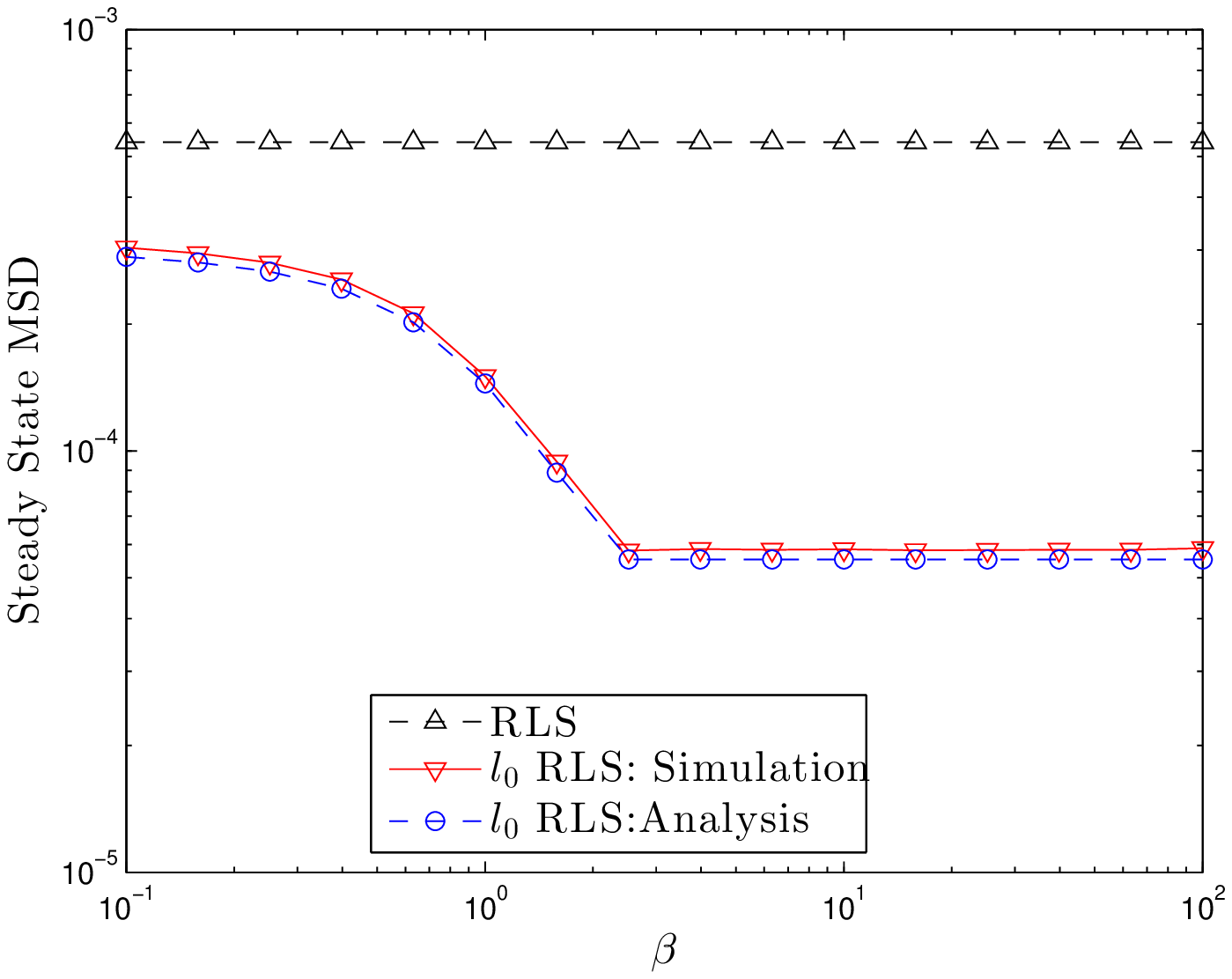}
\caption{Steady state MSD vs $\beta$ for SNR 25dB}
\label{fig:exp2-25dB}
\end{figure}
\begin{figure}[t!]
\includegraphics[height=3in,width=6in]{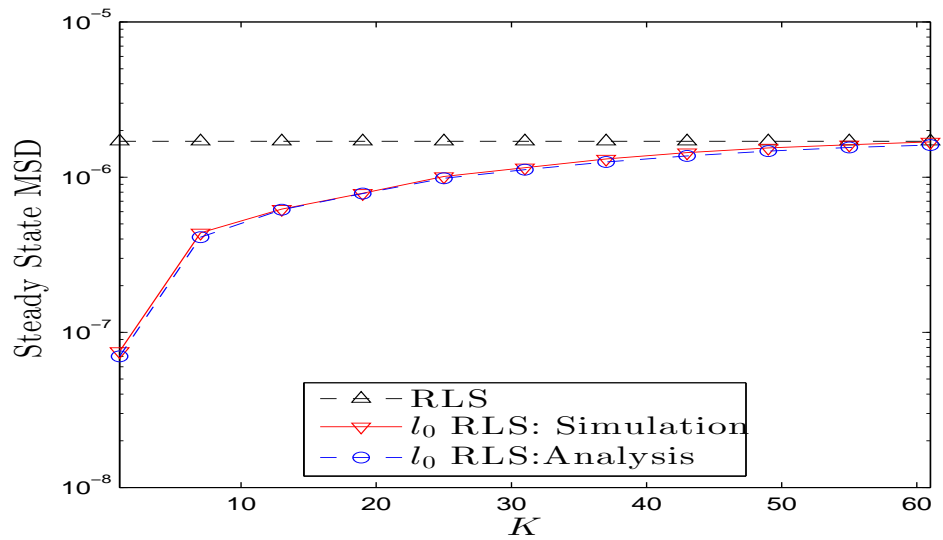}
\caption{Steady state MSD vs sparsity $K$ for SNR 50dB}
\label{fig:exp3-50dB}
\end{figure}
\section{Conclusion}
In this paper a theoretical analysis of $l_0$ RLS is carried out. Inspired by the work in~\cite{su2012performance}, relevant common assumptions are taken along with some new ones and their applicability are discussed. Also the taps are divided into different sets according to their magnitudes and the effect of the set with coefficients with small magnitude is analyzed in detail. The expressions for steady state MSD as well as a linear evolution model of the instantaneous MSD are derived and analyzed for the effects of different parameter settings. Several numerical simulations are done to verify the claims made by the analysis and are seen to match well with the theoretical predictions for a 
range of parameter values. 
\appendices
\section{Proof of Lemma~\ref{lem:ergodic-type}}
\label{sec:appendix-ergodic-type}
When $\lambda=1$, the proof follows from the ergodicity of the $\{x(n)\}$ sequence. 

For $\lambda\in(0,1)$, note that from Equation~(\ref{eq:inverse-define}) one can write $$\boldsymbol{\epsilon}_n:=\bvec{\Phi}_n-\frac{1-\lambda^{n+1}}{1-\lambda}\bvec{R}=\sum_{m=0}^n \lambda^{n-m}\left(\bvec{x}_m\bvec{x}_m^T-\bvec{R}\right)$$ Then, \begin{align*}
\expect\left(\boldsymbol{\epsilon}_n\boldsymbol{\epsilon}_n^T\right)=&\sum_{l,m}\lambda^{2n-l-m}\left(\expect\left(\bvec{x}_l\bvec{x}_l^T
\bvec{x}_m\bvec{x}_m^T\right)-\expect(\bvec{x}_l\bvec{x}_l^T) \bvec{R}-\bvec{R}\expect(\bvec{x}_m\bvec{x}_m^T)+\bvec{R}^2\right)\\
\ =&\sum_{l,m}\lambda^{2n-l-m}\left(\expect\left(\bvec{x}_l\bvec{x}_l^T
\bvec{x}_m\bvec{x}_m^T\right)-\bvec{R}^2\right)
\end{align*}
Now, using \emph{Gaussian mean factoring theorem}~\cite{priestley1981spectral}, one can find an expression for the $(i,j)^{\mathrm{th}}$ element ($0\le i,j\le N-1$)of $\expect\left(\bvec{x}_l\bvec{x}_l^T
\bvec{x}_m\bvec{x}_m^T\right)$ in the following way:\begin{align*}
\lefteqn{\left[\expect\left(\bvec{x}_l\bvec{x}_l^T
\bvec{x}_m\bvec{x}_m^T\right)\right]_{(i,j)}} & \\
\ &= \sum_{n=0}^{N-1} \expect(x(l-n)x(m-n)x(l-i)x(m-j))\\
\ &= \sum_{n=0}^{N-1} \left[\expect(x(l-n)x(m-n))\expect(x(l-i)x(m-j))+\expect(x(l-n)x(m-j))\expect(x(m-n)x(l-j))\right.\\
\ &+ \left.\expect(x(l-n)x(l-i))\expect(x(m-n)x(m-j)) \right]\\
\ &= \sum_{n=0}^{N-1}\sigpow^2\left(\delta(l-m)\delta(l-m-i+j)+\delta(l-m-n+j)\delta(m-n-l+j)+\delta(i-n)\delta(j-n)\right)\quad(\mbox{using assumption~\ref{assumption:data}})\\
 &= \left\{\begin{array}{ll}
(N+2)\sigpow^2, & l=m,\ i=j\\
\sigpow^2, & l\ne m,\ i=j\\
0, & \mathrm{otherwise}
\end{array}\right.
\end{align*}
Then, recalling that under assumption~\ref{assumption:data}, $\bvec{R}=\sigpow \bvec{I}$, we get \begin{align*}
\boldsymbol{\epsilon}_n\boldsymbol{\epsilon}_n^T=&(N+1)\sigpow^2\bvec{I}\sum_{m=0}^{n}\lambda^{2n-2m}\\
 =&\frac{(1-\lambda^{2(n+1)})(N+1)\sigpow^2}{1-\lambda^2}\bvec{I}
\end{align*}
So, $\forall \bvec{u}\in \real^N$ such that $\norm{\bvec{u}}=1$, we have \begin{align*}
\lefteqn{\norm{\frac{\boldsymbol{\epsilon}_n}{\frac{1-\lambda^{n+1}}{1-\lambda}}\bvec{u}}^2} & \\
\ & = \left(\frac{1-\lambda}{1+\lambda}\right)\left(\frac{1+\lambda^{n+1}}{1-\lambda^{n+1}}\right)(N+1)\sigpow^2\to \frac{1-\lambda}{1+\lambda}(N+1)\sigpow^2
\end{align*}
as $n\to \infty$. This proves the claim.
\section{Proof of Theorem~\ref{thm:mean-convergence}}
\label{sec:appendix-mean-convergence}
Taking expectations on both sides of Eq.~(\ref{eq:simplified-evolution}), and using Assumptions~\ref{assumption:data} and~\ref{assumption:independence}, we get \begin{align}
\label{eq:mean-evolve-instantaneous}
\expect\bvec{h}_n=\eta_n\expect\bvec{h}_{n-1}+\rho_n\expect\bvec{g}(\bvec{w}_{n-1})
\end{align}
To solve the linear system in Eq.~(\ref{eq:mean-evolve-instantaneous}), an expression for $\expect\bvec{g}(\bvec{w}_{n-1})$ is needed. Using the assumptions~\ref{assumption:zero-tap-weight},~\ref{assumption:tap-sign}, and~\ref{assumption:attraction}, we get, \begin{align}
\label{eq:g(w)-evaluation}
g(w_{k,n-1})=\left\{\begin{array}{lr}
0 & \forall\ k\in \largeset\\
\beta^2{h}_{k,n-1}+g(s_k) & \forall\ k\in \smallset\\
g(h_{k,n-1}) & \forall\ k\in \zeroset
\end{array}\right.
\end{align}
Thus, from Eq.~(\ref{eq:mean-evolve-instantaneous}) it follows that \begin{align}
\label{eq:mean-evolve-instantenous-further-simplified}
\expect h_{k,n}=&\left\{\begin{array}{lr}
\eta_n\expect h_{k,n-1} & \forall\ k\in \largeset\\
\eta_n\expect h_{k,n-1}+\rho_n g(s_k) & \forall\ k\in \smallset\\
\eta_n\expect h_{k,n-1} & \forall\ k\in \zeroset
\end{array}\right.
\end{align}
where in Eq.~(\ref{eq:mean-evolve-instantenous-further-simplified}) the assumption~\ref{assumption:kappa-param} is used to simplify the expression for $\expect g(w_{k,n-1})$ for $k\in \smallset$. The expression for $\expect g(w_{k,n-1})$ for $k\in \zeroset$ is obtained in the following way, using assumption~\ref{assumption:zero-tap-weight} and the definition of function $g(\cdot)$ in Eq.~(\ref{eq:define-g-function}):  \begin{align*}
\expect{g(w_{k,n-1})}=&\frac{1}{\sqrt{2\pi\omega_{n-1}^2}}\int_{-1/\beta}^{1/\beta}(\beta^2x-\mathrm{sgn}(x))e^{-x^2/2\omega_{n-1}^2}dx\\ 
\ =&0
\end{align*} 
where $\omega_n^2:=\expect{h^2_{k,n}}\ \forall k\in \zeroset$. 
Then, it follows that $\forall k\in \smallset $, \begin{align*}
\expect h_{k,n}=\left\{\begin{array}{ll}
\prod_{k=1}^n \eta_k \expect h_{k,0}+\left(\sum_{k=1}^n \rho_k\prod_{j=k+1}^{n}\eta_j \right) g(s_k), & k\in \smallset\\
\prod_{k=1}^n \eta_k \expect h_{k,0}, & k\in \zeroset \cup \largeset
\end{array} 
\right.
\end{align*}
From definitions of $\eta_n,\ \rho_n$, we find that $\prod_{k=1}^n \eta_k=\lambda^n\frac{1-\lambda}{1-\lambda^{n+1}}=c_n$, and \begin{align*}
\sum_{k=1}^n \rho_k\prod_{j=k+1}^{n}\eta_j=&\sum_{k=1}^n \left(\frac{\kappa}{\sigpow}\frac{1-\lambda}{1-\lambda^{k+1}}\right)\cdot \left(\lambda^{n-k}\frac{1-\lambda^{k+1}}{1-\lambda^{n+1}}\right)\\
\ =&\rho_n\sum_{k=1}^n \lambda^{n-k}\\
\ =&\frac{\kappa}{\sigpow}\frac{1-\lambda^n}{1-\lambda^{n+1}}=d_n
\end{align*}
From this the evolution equation for $\expect h_{k,n}$ follows. Taking, $n\to \infty$ trivially results in Eq.~(\ref{eq:mean-convergence-equation}).
\section{Proof of Theorem~\ref{thm:evolve-instantaneous-zero-nonzero-tap-pow}}
To solve the recursion in Eq.(~\ref{eq:evolve-diag-error-covariance-mat}), we need to evaluate the terms $\expect(h_{k,n-1}g(w_{k,n-1})),\ \expect(g^2(w_{k,n-1}))$, and~$\expect(k_{k,n}^2)$ for each $k\in \{1,2,\cdots,\ N\}$.
\label{sec:appendix-evolve-instantaneous-zero-nonzero-tap-pow}
\subsection{\textbf{Evaluating $\expect(h_{k,n-1}g(w_{k,n-1}))$}}
 From Eq.~(\ref{eq:g(w)-evaluation}) and recalling that $h_{k,n}=w_{k,n}-s_k$, we get 
 \begin{align}
\label{eq:hg(w)-evaluation}
h_{k,n-1}g(w_{k,n-1})=\left\{\begin{array}{lr}
0 & \forall\ k\in \largeset\\
\beta^2{h^2}_{k,n-1}+g(s_k)h_{k,n-1} & \forall\ k\in \smallset\\
h_{k,n-1}g(h_{k,n-1}) & \forall\ k\in \zeroset
\end{array}\right.
\end{align}
So, taking expectations on both sides we get
 \begin{align*}
\expect h_{k,n-1}g(w_{k,n-1})=\left\{\begin{array}{lr}
0 & \forall\ k\in \largeset\\
\beta^2\expect {h^2}_{k,n-1}+g(s_k)\expect h_{k,n-1} & \forall\ k\in \smallset\\
\expect h_{k,n-1}g(h_{k,n-1}) & \forall\ k\in \zeroset
\end{array}\right.
\end{align*}
To get the expression for $\expect h_{k,n-1}g(w_{k,n-1})$ for $k\in \zeroset$, we note that, for $k\in \zeroset$, using the definition of function $g(\cdot)$ in Eq.(~\ref{eq:define-g-function}), we get \begin{align*}
\expect h_{k,n-1}g(h_{k,n-1})=&\frac{1}{\sqrt{2\pi\omega_{n-1}^2}}\int_{-1/\beta}^{1/\beta}(\beta^2 x^2-\beta|x|)e^{-x^2/2\omega_{n-1}^2}dx\\
\end{align*}
Note that assumption~\ref{assumption:zero-tap-weight} implies that $\omega_n<1/\beta,\ \forall n\ge 1$, which permits to approximate the above integral as \begin{align*}
\expect h_{k,n-1}g(h_{k,n-1})=&\frac{1}{\sqrt{2\pi}}\int_{-\infty}^{\infty}(\beta^2 \omega_{n-1}^2x^2-\omega_{n-1}\beta|x|)e^{-x^2/2}dx\\
\ =&\beta^2\omega_{n-1}^2-\frac{2\beta\omega_{n-1}}{\sqrt{2\pi}}
\end{align*}

Thus  \begin{align}
\label{eq:expect-hg(w)-evaluation}
\expect h_{k,n-1}g(w_{k,n-1})=\left\{\begin{array}{lr}
0 & \forall\ k\in \largeset\\
\beta^2\expect {h^2}_{k,n-1}+g(s_k)\expect h_{k,n-1} & \forall\ k\in \smallset\\
\beta^2\omega_{n-1}^2-\frac{2\beta\omega_{n-1}}{\sqrt{2\pi}} & \forall\ k\in \zeroset
\end{array}\right.
\end{align}
where $\omega_{n-1}^2:=\expect h_{k,n-1}^2,\ k\in \zeroset$.

\subsection{\textbf{Evaluating $\expect g^2(w_{k,n-1})$}} 
Again using the definition of function $g(\cdot)$ in Eq.(~\ref{eq:define-g-function}), we get  \begin{align*}
\expect g^2(w_{k,n-1})=\left\{\begin{array}{lr}
0 & \forall\ k\in \largeset\\
\beta^4\expect {h^2}_{k,n-1}+g^2(s_k)+2\beta^2g(s_k)\expect(h_{k,n-1}) & \forall\ k\in \smallset\\
\expect g^2(h_{k,n-1}) & \forall\ k\in \zeroset
\end{array}\right.
\end{align*}
Again, using assumption~\ref{assumption:zero-tap-weight}, we get, $\forall k\in \zeroset$, \begin{align*}
\expect g^2(h_{k,n-1})=&\frac{1}{\sqrt{2\pi}}\int_{-\infty}^{\infty}(\beta^2 \omega_{n-1}x-\beta\mathrm{sgn}(x))^2e^{-x^2/2}dx\\
\ =&\beta^2-2\beta^3\omega_{n-1}\frac{1}{\sqrt{2\pi}}\int_{-\infty}^\infty |x|e^{-x^2/2}dx+\beta^4\omega_{n-1}^2\frac{1}{\sqrt{2\pi}}\int_{-\infty}^\infty x^2e^{-x^2/2}dx\\ 
\ =&\beta^2-\frac{4\beta^3\omega_{n-1}}{\sqrt{2\pi}}+\beta^4\omega_{n-1}^2
\end{align*}
Thus, \begin{align}
\label{eq:expect-g^2(w)-evaluation}
\expect g^2(w_{k,n-1})=\left\{\begin{array}{lr}
0 & \forall\ k\in \largeset\\
\beta^4\expect {h^2}_{k,n-1}+g^2(s_k)+2\beta^2g(s_k)\expect(h_{k,n-1}) & \forall\ k\in \smallset\\
\beta^2-\frac{4\beta^3\omega_{n-1}}{\sqrt{2\pi}}+\beta^4\omega_{n-1}^2 & \forall\ k\in \zeroset
\end{array}\right.
\end{align}
\subsection{\textbf{Evaluating $\expect k_{k,n}^2$}}
From the definition of the \emph{gain vector} $\bvec{k}_n$ in Eq.(~\ref{eq:gain-mat-define}), along with the assumptions~\ref{assumption:data} and~\ref{assumption:inverse-simplification}, we get the following simplified expression for $\bvec{k}_n$: \begin{align}
\label{eq:gain-mat-simplified}
\bvec{k}_n=\frac{\bvec{x}_n}{a_n^2+\norm{\bvec{x}_n}^2}
\end{align}
where $a_n^2:=(1-\lambda^n)a^2$ and $a^2:=\frac{\lambda}{1-\lambda}\sigpow$. Then, \begin{align*}
k_{k,n}=\frac{{x}_{k,n}}{a_n^2+\norm{\bvec{x}_n}^2}
\end{align*}
Let, $p_{k,n}^2:=\expect k_{k,n}^2$. It follows from assumption~\ref{assumption:data} that $p_{0,n}=p_{1,n}=\cdots=p_{N-1,n}=p_n$ where \begin{align*}
p_n^2=\frac{1}{N}\expect\left(\frac{\norm{\bvec{x}_n}^2}{(a_n^2+\norm{\bvec{x}_n}^2)^2}\right)
\end{align*}
Now, because of the choice of $\lambda$ in assumption~\ref{assumption:inverse-simplification}, we have $\frac{\lambda}{1-\lambda}>>N$. Then, we can simplify the expression for $p_n$ as an approximation \begin{align}
p_n^2\approx &\frac{1}{N}\expect\left(\frac{\norm{\bvec{x}_n}^2}{(a_n^2)^2}\right)\nonumber\\
\ =&\frac{\sigpow}{a_n^4}\nonumber\\
\label{eq:approx-pn}
\implies p_n^2=&\frac{(1-\lambda)^2}{\lambda^2(1-\lambda^n)^2 \sigpow}\\
\label{eq:p-infinity-approx}
p_\infty^2=&\frac{(1-\lambda)^2}{\lambda^2 \sigpow}
\end{align}  
\subsection{\textbf{Putting everything together}}
Thus, using the expressions found in equations~\ref{eq:expect-hg(w)-evaluation},~\ref{eq:expect-g^2(w)-evaluation}, and~\ref{eq:gain-mat-simplified} in Eq.~(\ref{eq:evolve-diag-error-covariance-mat}) and using the assumption~\ref{assumption:kappa-param}, we get \begin{align}
\label{eq:evolve-mean-square-intermediate}
\expect h_{k,n}^2=&\left\{\begin{array}{lr}
\eta_n^2\expect h_{k,n-1}^2+\noisepow p_n^2, & k\in \largeset\\
\eta_n^2\expect h_{k,n-1}^2+\noisepow p_n^2+2\rho_n\eta_n g(s_k)\expect h_{k,n-1}+\rho_n^2 g^2(s_k), & k\in \smallset\\
\eta_n^2\expect h_{k,n-1}^2+\noisepow p_n^2-4\beta\rho_n\eta_n\frac{\sqrt{\expect h^2_{k,n-1}}}{\sqrt{2\pi}}+\beta^2\rho_n^2, & k\in \zeroset
\end{array}\right.
\end{align} 

This along with Eq.~(\ref{eq:mean-evolution-solved}) produces the following linear recursion:\begin{align}
\label{eq:evolve-mean-square-final}
\expect h_{k,n}^2=&\left\{\begin{array}{lr}
\eta_n^2\expect h_{k,n-1}^2+\noisepow p_n^2, & k\in \largeset\\
\eta_n^2\expect h_{k,n-1}^2+\noisepow p_n^2-2\rho_n c_n\eta_n s_k g(s_k)+2\rho_n d_n\eta_n g^2(s_k)+\rho_n^2 g^2(s_k), & k\in \smallset\\
\eta_n^2\expect h_{k,n-1}^2+\noisepow p_n^2-4\beta\rho_n\eta_n\frac{\sqrt{\expect h^2_{k,n-1}}}{\sqrt{2\pi}}+\beta^2\rho_n^2, & k\in \zeroset
\end{array}\right.
\end{align}
where we have assumed that $\expect w_{0,k}=0,\ \forall k$.
  Then, it follows from Eq.~(\ref{eq:evolve-mean-square-final}) \begin{align}
\label{eq:evolve-pow-nonzero-tap}
D_n-\Omega_n=&\eta_n^2(D_{n-1}-\Omega_{n-1})+K\noisepow p_n^2-2\rho_n c_n\eta_n G'(s)+(2\rho_n d_n\eta_n+\rho_n^2)G(s)
\end{align} 
Also, it follows from Eq.~(\ref{eq:evolve-mean-square-final}) \begin{align*}
\Omega_n=\eta_n^2\Omega_{n-1}+(N-K)(\beta^2\rho_n^2+\noisepow p_n^2)-(N-K)\frac{4\beta\rho_n\eta_n}{\sqrt{2\pi}}\omega_{n-1}
\end{align*}
Observing that $\Omega_n=(N-K)\omega_n^2$, it follows that \begin{align}
\label{eq:evolve-pow-zerotap}
\Omega_n=\eta_n^2\Omega_{n-1}+(N-K)(\beta^2\rho_n^2+\noisepow p_n^2)-\sqrt{N-K}\frac{4\beta\rho_n\eta_n}{\sqrt{2\pi}}\sqrt{\Omega_{n-1}}
\end{align}

Thus, using $\Omega_n=(N-K)\omega_n^2,\ \forall n$, it follows from Eq.~(\ref{eq:evolve-pow-zerotap}), as $n\to \infty$,\begin{align*}
\omega_{\infty}^2=&\eta_{\infty}^2\omega_{\infty}^2+\beta^2\rho_{\infty}^2+\noisepow p_\infty^2-\frac{4\beta\rho_{\infty}\eta_{\infty}}{\sqrt{2\pi}}\omega_{\infty}\nonumber\\
\implies\omega_{\infty}=&\frac{-2\beta\rho_{\infty}\eta_{\infty}/\sqrt{2\pi}+
\sqrt{2\beta^2\rho^2_{\infty}\eta^2_{\infty}/\pi+(1-\eta_{\infty}^2)(\beta^2\rho_{\infty}^2+\noisepow p_\infty^2)}}{(1-\eta_{\infty}^2)}\quad(\because \ \omega_\infty\ge 0)
\end{align*}
Now, $\eta_{\infty}=\lambda,\ \rho_\infty=\frac{\kappa(1-\lambda)}{\sigpow}$. recalling Hence, we have the desired parametric expression for $\omega_\infty$ in terms of $\theta$ as promised in Theorem~\ref{thm:evolve-instantaneous-zero-nonzero-tap-pow}.

For large $n$, however, an approximate linear evolution for $\Omega_n$ can be obtained by a first order Taylor series approximation of $\sqrt{\expect h^2_{k,n-1}}$ to get \begin{align*}
{\expect h^2_{k,n-1}}\approx & \sqrt{\expect h^2_{k,\infty}}+\frac{\expect h^2_{k,n-1}-\expect h^2_{k,\infty}}{2\sqrt{\expect h^2_{k,\infty}}}\\
\ =& \frac{\expect h^2_{k,n-1}+\expect h^2_{k,\infty}}{2\sqrt{\expect h^2_{k,\infty}}}
\end{align*} 
then Eq.~(\ref{eq:evolve-pow-zerotap}) will become \begin{align}
\label{eq:evolve-pow-zerotap-linear-approx}
\Omega_n=\eta_n^2\Omega_{n-1}+(N-K)(\beta^2\rho_n^2+\noisepow p_n^2)-\frac{2\beta\rho_n\eta_n}{\sqrt{2\pi\omega_{\infty}^2}}(\Omega_{n-1}+(N-K)\omega_{\infty}^2)
\end{align}
For large $n$, thus, Eq.~(\ref{eq:evolve-pow-nonzero-tap}) and Eq.~(\ref{eq:evolve-pow-zerotap-linear-approx}) together produce the results in Equations~(\ref{eq:evolve-instantaneous-zero-nonzero-tap-pow}) and~(\ref{eq:instantaneous-evolve-matrix}).
\section{Proof of Theorem~\ref{thm:steady-msd}}
\label{sec:appendix-steady-msd}
From Eq.~(\ref{eq:evolve-pow-nonzero-tap}), taking $n\to \infty$ and using Eq.~(\ref{eq:steady-zerotap}), we get \begin{align*}
D_\infty=\frac{(N-K)\beta^2\rho_\infty^2+N\noisepow p_\infty^2-2\rho_\infty c_\infty\eta_\infty G'(s)+(2\rho_\infty d_\infty\eta_\infty+\rho_\infty^2)G(s)-(N-K)\frac{4\beta\rho_\infty\eta_\infty}{\sqrt{2\pi}}\omega_\infty}{1-\eta_\infty^2}
\end{align*}Observing that \begin{align*}
\rho_\infty=& \frac{\kappa(1-\lambda)}{\sigpow}\\
\eta_\infty=& \lambda\\
c_\infty=& 0\\
d_\infty=& \frac{\kappa}{\sigpow}
\end{align*}
we get \begin{align}
D_\infty=\frac{(N-K)\theta^2+N\noisepow p_\infty^2+(\frac{2\lambda\theta^2}{\beta^2(1-\lambda)}+\frac{\theta^2}{\beta^2})G(s)-\frac{4\lambda(N-K)\theta}{\sqrt{2\pi}}\omega_\infty}{1-\lambda^2}
\end{align}
which, together with Eq.~(\ref{eq:steady-zerotap}), yields the desired result.
%
%
\bibliography{l0_rls_analysis}
\end{document}